\documentclass[a4paper,english,numberwithinsect]{eurocg20}
\usepackage{graphicx,xspace,amsmath,enumitem}
\usepackage{amsthm}
\usepackage{dsfont}
\usepackage{booktabs}
\usepackage{xcolor}
\usepackage[nofancy]{svninfo}
\svnInfo $Id: newArcs.tex 7566 2018-10-25 15:48:32Z eth.foersth $

\usepackage{todonotes}

\theoremstyle{theorem}

\usepackage{amsfonts,amssymb,bm,amsthm,thm-restate}

\newcommand{\OC}[1]{$\mathcal{C}_{#1}$} %
\newcommand{\N}{\ensuremath{\mathds N}} %
\graphicspath{{figures/}}

\begin{document}

  \title{Monotone Arc Diagrams with few Biarcs\footnote{This work started at the
      workshop \emph{Graph and Network Visualization} 2017. We would like to
      thank Stefan Felsner and Stephen Kobourov for useful
      discussions.}} 

\date{Revision \svnInfoRevision\ --- \svnToday}
\authorrunning{Steven~Chaplick, Henry F\"orster, Michael Hoffmann and Michael Kaufmann} 

\author[1]{Steven Chaplick\thanks{Partially supported by DFG grant WO 758/11-1.}}
\author[2]{Henry F\"orster}
\author[3]{Michael Hoffmann\thanks{Supported by the Swiss National
    Science Foundation within the collaborative DACH project \emph{Arrangements
      and Drawings} as SNSF Project 200021E-171681.}}
\author[2]{Michael Kaufmann}
\affil[1]{Universit\"at W\"urzburg, Germany and Maastricht University, the Netherlands\\\texttt{s.chaplick@maastrichtuniversity.nl}} 
\affil[2]{Universit\"{a}t T\"{u}bingen, Germany\\\texttt{\{foersth,mk\}@informatik.uni-tuebingen.de}} 
\affil[3]{Department of Computer Science, ETH Z\"urich, Switzerland\\\texttt{hoffmann@inf.ethz.ch}} 

\captionsetup[subfigure]{justification=centering}


\maketitle

\begin{abstract}
  We show that every planar graph 
  can be represented by a monotone topological $2$-page book embedding where at
  most $15n/16$ (of potentially $3n-6$) edges cross the spine exactly once.
\end{abstract}

\section{Introduction}\label{sec:intro}

\emph{Arc diagrams} (\figurename~\ref{fig:0}) are drawings of graphs that represent
vertices as points on a horizontal line, called \emph{spine}, and edges as \emph{arcs}, consisting of a sequence of
halfcircles centered on the spine. A \emph{proper arc} consists of one single
halfcircle. In \emph{proper arc diagrams} all arcs are proper. In  
\emph{plane} arc diagrams no two edges cross. Note that plane proper arc diagrams are also known as \emph{$2$-page book embeddings} in the literature. Bernhard and Kainen~\cite{bk-btg-79}
characterized the graphs admitting plane proper arc diagrams:
subhamiltonian planar graphs, i.e., subgraphs of planar
graphs with a Hamiltonian cycle. In particular, non-Hamiltonian maximal planar graphs do not
admit plane proper arc diagrams. 

\begin{figure}[bhtp]
  \centering%
  \begin{subfigure}[b]{.3\textwidth}
  \centering
	\includegraphics{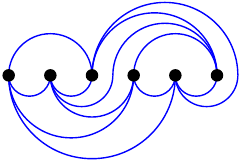} 
	\subcaption{}	
	\label{fig:1:a}  
  \end{subfigure}\hfill
  \begin{subfigure}[b]{.3\textwidth}
  \centering
 	\includegraphics{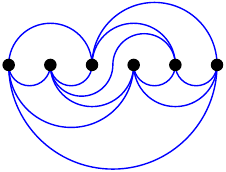}
  	\subcaption{}	
	\label{fig:1:b}  
  \end{subfigure}\hfill
  \begin{subfigure}[b]{.3\textwidth}
  \centering
  	\includegraphics{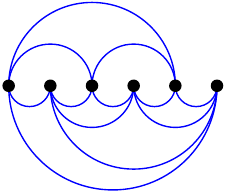}
  	\subcaption{}	
	\label{fig:1:c}
  \end{subfigure} 
  \caption{Arc diagram (a), monotone arc
    diagram (b), proper arc diagram (c) of the octahedron.\label{fig:0}}
\end{figure}

To represent all planar graphs, it suffices to allow each edge to cross the
spine at most once~\cite{kw-evpfb-02}. The resulting arcs  composed of two
halfcircles are called \emph{biarcs} (see \figurename~\ref{fig:1:a}).
Additionally, all edges can be drawn as \emph{monotone} curves w.r.t.  the
spine~\cite{ddlw-ccdpg-05}; such a drawing is called a \emph{monotone
  topological ($2$-page) book embedding}. A monotone biarc is either
\emph{down-up} or \emph{up-down}, depending on if the left halfcircle is drawn
above or below the spine, respectively. Note that a \emph{monotone topological} $2$-page book embedding is not necessarily a $2$-page book embedding even though the terminology suggests it.

In general, biarcs are needed, but \emph{some} edges can be drawn as proper
arcs. Cardinal et al{.}~\cite{chktw-adfdht-18} gave bounds on the required
number of biarcs showing that every planar graph on $n\ge 3$ vertices admits a
plane arc diagram with at most $\lfloor(n-3)/2\rfloor$ biarcs (not necessarily
monotone). They also described a family of planar graphs on $n_i=3i+8$ vertices
that cannot be drawn as a plane biarc diagram using less than $(n_i-8)/3$ biarcs
for $i\in\N$. However, they use arbitrary biarcs. 
When requiring only monotone biarcs,
Di Giacomo et al.~\cite{ddlw-ccdpg-05} gave an algorithm to
construct a monotone plane arc diagram that may create close to $2n$ biarcs for
an $n$-vertex planar graph. Cardinal et al{.}~\cite{chktw-adfdht-18} improved
this bound to at most $n-4$ biarcs.

\subparagraph*{Results.}
As a main result, we improve the upper bound on the number of monotone biarcs:

\begin{restatable}{theorem}{mainthm}\label{thm:1}
  Every $n$-vertex planar graph admits a plane arc diagram with at
  most $\left\lfloor \frac{15}{16}n- \frac{5}{2} \right\rfloor $ biarcs that are all down-up monotone. Such a diagram is computable in $O(n)$~time.
\end{restatable}

For general arc diagrams, $\lfloor (n-8)/3\rfloor$ biarcs may be needed~\cite{chktw-adfdht-18}, but it is
 conceivable that this number increases for monotone biarcs. We investigated the
 lower bound with a SAT based approach (based on~\cite{DBLP:conf/gd/Bekos0Z15}),
with the following partial result; details will appear in the full version.
\begin{restatable}{observation}{satResult}
 Every Kleetope on $n'=3n-4$ vertices
derived from triangulations of  $n\le 14$ vertices admits a plane arc diagram with
$\lfloor (n'-8)/3 \rfloor$ monotone biarcs.
 \end{restatable}
 
Note that a \emph{Kleetope} is derived from a planar triangulation $T$ by inserting a new vertex $v_f$ into each face $f$ of $T$ and then connecting $v_f$ to the three vertices bounding $f$.
%

\subparagraph*{Related Work.}

Giordano et al{.}~\cite{GiordanoLMSW15} showed that every upward planar graph
admits an \emph{upward topological book embedding} where edges are either
proper arcs or biarcs. One of their directions for future work is to 
minimize the number of spine crossings. Note that these embeddings are monotone arc
diagrams with at most one spine
crossing per edge respecting the orientations of the edges. 
Everett et al.~\cite{DBLP:journals/dcg/EverettLLW10} used
monotone arc diagrams with only down-up biarcs to construct small universal point sets for $1$-bend
drawings of planar graphs. This result was
extended by L\"offler and T\'oth~\cite{LofflerT15} by restricting the set of
possible bend positions. They use  monotone arc diagrams with at
most $n-4$ biarcs to build universal points set of size $6n-10$ (vertices and
bend points) for $1$-bend drawings of planar graphs on $n$ vertices. Using Theorem~\ref{thm:1}, we can slightly decrease the number of points by approximately $n/16$.

\section{Overview of our Algorithm}\label{sec:overview}

To prove Theorem~\ref{thm:1} we describe an algorithm to incrementally construct
an arc diagram for a given planar graph $G$ on $n$ vertices. W.l.o.g. we assume that $G$ is a (combinatorial) \emph{triangulation},
i.e., a maximal planar graph.
Our algorithm is a (substantial) refinement of the algorithm of Cardinal et
al{.}, which  is based on the notion of a canonical ordering. A canonical
ordering is defined for an \emph{embedded} triangulation. Every triangulation
on $n\ge 4$ vertices is $3$-connected, so selecting one facial triangle as the \emph{outer face} embeds it into~the plane which determines a unique outer face (cycle) for every biconnected subgraph.
A \emph{cano\-nical ordering}~\cite{fpp-hdpgg-90} of an embedded triangulation
$G$  is a total order of vertices $v_1,\ldots,v_n$~s.t.
\begin{itemize}
\item[--] for each $i\in\{3,\ldots,n\}$, the induced subgraph $G_i=G[\{v_1,\ldots,
  v_i\}]$ is biconnected and internally triangulated (i.e., every inner face 
  is a triangle);
\item[--] for each $i\in\{3,\ldots,n\}$, $(v_1,v_2)$ is an edge of the outer face $C_i$ 
of $G_i$;
\item[--] for each $i\in\{3,\ldots,n-1\}$, $v_{i+1}$ lies in the interior of $C_i$  and the neighbors of $v_{i+1}$
  in $G_i$ form a sequence of consecutive vertices along the boundary of $C_i$. 
\end{itemize}
Every triangulation admits a canonical ordering~\cite{fpp-hdpgg-90} and one can be computed in $O(n)$ time~\cite{cp-ltadp-95}.  We say that a vertex $v_i$
\emph{covers} an edge $e$ (a vertex $v$, resp.) if and only if $e$ ($v$, resp.) is an edge
(vertex, resp.) on $C_{i-1}$ but not an edge (vertex, resp.) on $C_i$.

We iteratively process the vertices in a canonical order $v_1,\ldots,v_n$.
Every vertex $v_i$ arrives with $\alpha$ credits that we can either spend to
create biarcs (at a cost of one credit per biarc) or distribute on edges of
the outer face $C_i$ for later use. 
We prove our claimed bound by showing that each biarc drawn 
 can be paid for s.t. at least seven credits remain in total. 

There are two types of proper arcs: \emph{mountains} (above
the spine) and \emph{pockets} (below the spine). The following invariants hold after processing vertex $v_i$, for every
$i\in\{3,\ldots,n\}$.

\begin{enumerate}[label={(I\arabic*)}]\setlength{\itemindent}{\labelsep}
\item\label{i:biarcTypes} Every edge is either a proper arc or a down-up biarc.
\item\label{i:contour} Every edge of $C_i$ is a proper arc. 
  Vertex $v_1$ is the leftmost and $v_2$ is the rightmost vertex of $G_i$. Edge
  $(v_1,v_2)$ forms the lower envelope of $G_i$, i.e., no point of the drawing
  is vertically below it. The other edges of $C_i$ form the upper envelope of
  $G_i$.
\item\label{i:mountainMoney} Every mountain whose left endpoint is on $C_i$
  carries $1$ credit.
\item\label{i:pocketMoney} Every pocket on $C_i$ carries $\pi$ credits, for some
  constant $\pi \in (0,1)$.
\item\label{i:biarcMoney} Every biarc in $G_i$ carries (that is, is paid for
  with) $1$ credit.
\end{enumerate}

Usually, we insert $v_i$ between its leftmost neighbor $\ell_i$ and 
rightmost neighbor $r_i$ along $C_{i-1}$. The algorithm of Cardinal et
al{.}~\cite{chktw-adfdht-18} gives a first upper bound on the insertion
costs.

\begin{restatable}{lemma}{defaultApproachOne}
  If $v_i$ covers at least one pocket, then we can insert $v_i$ maintaining
  \ref{i:biarcTypes} to~\ref{i:biarcMoney} using 
  $\le 1$ credit. If $\deg_{G_{i}}(v_i) \geq 4$, then $1-\pi$ credits are
  enough.
\label{lem:defaultApproach1}
\end{restatable}

\begin{proof}[Proof (Sketch)]
We place $v_i$ in the rightmost covered pocket and  pay for at most $1$ mountain; see Figures~\ref{fig:naivePocket2} and~\ref{fig:naivePocket}. If $\deg_{G_{i}}(v_i) \geq 4$, at least $1$ covered pocket's credits is free.
\end{proof}

\begin{figure}[htbp]
  \centering%
    \begin{subfigure}[b]{.275\textwidth}
    \centering
    \includegraphics[page=52]{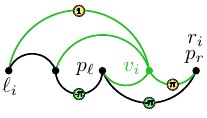}
	\subcaption{}     
    \label{fig:naivePocket2}
    \end{subfigure}\hfill
    \begin{subfigure}[b]{.35\textwidth}
    \centering
    \includegraphics[page=50]{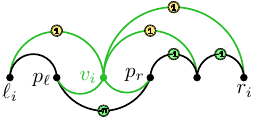}
    \subcaption{}  
    \label{fig:naivePocket}  
    \end{subfigure}\hfill
  \begin{subfigure}[b]{.275\textwidth}
   \centering
    \includegraphics[page=51]{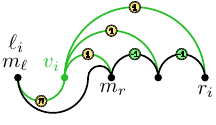}
    \subcaption{}  
    \label{fig:naiveMountains}
    \end{subfigure}
  \caption{Inserting a vertex $v_i$ using $1-\pi$, $1$, and $1+\pi$ credits,
    resp. (Lemma~\ref{lem:defaultApproach1}--\ref{lem:defaultApproach2}).}
\label{fig:naive}
\end{figure}

\begin{restatable}{lemma}{defaultApproachTwo}
  \label{lem:defaultApproach2}
  If $v_i$ covers mountains only, then we can insert $v_i$ maintaining
  \ref{i:biarcTypes} to~\ref{i:biarcMoney} using $\le 1+\pi$ credits. If
  $\deg_{G_{i}}(v_i) \geq 4$, then $5-\deg_{G_{i}}(v_i)$ credits suffice.%
\end{restatable}

 \begin{proof}[Proof (Sketch)]
 If $\deg_{G_{i}}(v_i) < 4$, we push down the leftmost mountain and place $v_i$ on the created biarc paying for $1$ mountain and $1$ pocket each; see Figure~\ref{fig:naiveMountains}. If
  $\deg_{G_{i}}(v_i) \geq 4$, we push down the rightmost mountain saving the credit of a covered mountain; see Figure~\ref{fig:degree4+Mountains}.  
 \end{proof}

\begin{figure}[hb]
  \centering%
  \includegraphics[page=56]{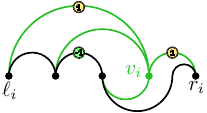}
  \caption{An alternative drawing to insert a degree four vertex.}
  \label{fig:degree4+Mountains}  
\end{figure}

Full proofs of Lemmas~\ref{lem:defaultApproach1} and~\ref{lem:defaultApproach2} will appear in the full version.
%
We only steal credits from arcs
on $C_{i-1}$ in both  proofs. If left endpoints of proper arcs not on $C_{i-1}$ are covered, there is slack. 

In the following, we prove that we can choose $\pi=1/8$, so that to achieve the bound of Theorem~\ref{thm:1}, we insert a vertex
at an average cost of $1-\pi/2$. Lemmas~\ref{lem:defaultApproach1}~and~\ref{lem:defaultApproach2}
guarantee this bound only in certain cases, e.g., a sequence of three degree
two (in $G_i$) vertices stacked onto mountains costs $1+\pi$ 
per vertex and produces three biarcs, see \figurename~\ref{fig:degree2stackfw}. A symmetric scheme
with up-down biarcs realizes the same graph with one biarc; see
\figurename~\ref{fig:degree2stackbw}.

\begin{figure}[htbp]
    \centering
  \begin{subfigure}[b]{.45\textwidth}
    \centering
    \includegraphics[page=53]{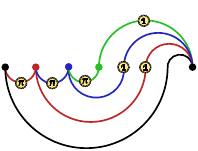}
    \subcaption{}  
    \label{fig:degree2stackfw}
    \end{subfigure}\hfill
  \begin{subfigure}[b]{.45\textwidth}
    \centering
    \includegraphics[page=54]{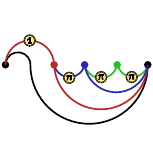}
    \subcaption{}  
    \label{fig:degree2stackbw}
    \end{subfigure}\hfill
  \caption{A
    sequence of degree two vertices in forward (a) and reverse drawing (b).}
\label{fig:problems}
\end{figure}

To exploit this behavior, we consider the instance in both a \emph{forward drawing}, using only
proper arcs and down-up biarcs, and a \emph{reverse drawing} that uses only
proper arcs and up-down biarcs (and so \ref{i:biarcTypes} and
\ref{i:mountainMoney} appear in a symmetric formulation). Out of the two
resulting arc diagrams, we choose one with a fewest number of biarcs.
To prove Theorem~\ref{thm:1}, we need to insert a vertex at an average
cost of $\alpha=2-\pi$ credits into both diagrams.

The outer face, a
sequence of pockets and mountains, can evolve differently in both
drawings because edges covered by a
vertex may not be drawn the same way in both drawings. Further, it
does not suffice to consider a single vertex in isolation. For instance,
consider a degree three vertex inserted above two mountains in both
the forward and reverse drawings; see \figurename~\ref{fig:openConfigurationB}. In
each drawing, this costs $1+\pi$ credits, or $2(1+\pi)$ in total. W.r.t. our target value $\alpha=2-\pi$, these costs incur a
\emph{debt} of $3\pi$ credits. Indeed, there are several such \emph{open
  configurations}, listed in \figurename~\ref{fig:openConfigurationsOverview},
for which our basic analysis does not suffice.


Each open configuration $\mathcal{C}$ consists of up to two adjacent vertices on
the outer face whose insertion incurred a debt and their incident edges. 
It specifies the drawing of these edges, as pocket, mountain, or biarc
in forward and in reverse drawing, as well as the drawing 
of the edges covered by the vertices of the open configuration.  When a vertex
$v_i$ covers (part of) an open configuration, we may alter the placement of the
vertices and/or draw the edges of the open configuration differently.
The associated debt $d(\mathcal{C})$ is the amount of credits paid in addition to
$\alpha$ credits per vertex.
As soon as any arc of an open configuration is covered, the debt must be paid or
transferred to a new open configuration. We enhance our collection of
invariants as follows.

\begin{enumerate}[label={(I\arabic*)},start=6]\setlength{\itemindent}{\labelsep}
\item\label{i:openConfigurationDebts} A sequence of consecutive arcs on $C_i$ may be associated with a debt. Each arc is part of at most one open configuration; refer to
  \figurename~\ref{fig:openConfigurationsOverview} for a full list of such configurations.
\end{enumerate}

To prove Theorem~\ref{thm:1} we show that the credit total carried by arcs in
both drawings minus the total debt of all open configurations does not exceed
$\alpha i-5$ after inserting $v_i$.

\section{Default insertion of a vertex $v_i$}
\label{sec:default}

If $v_i$ does not cover any arc of an open
configuration, we use procedures from
Lemmas~\ref{lem:defaultApproach1} and~\ref{lem:defaultApproach2}. If $\deg_{G_i}(v_i)\ge 4$ and $v_i$ covers any pocket in either drawing, by Lemmas~\ref{lem:defaultApproach1} and~\ref{lem:defaultApproach2} the insertion
costs are at most $2-\pi=\alpha$. If $\deg_{G_i}(v_i)\ge 5$ and $v_i$ only covers
mountains in both drawings, the
 costs are $0$ (Lemma~\ref{lem:defaultApproach2}). If $\deg_{G_i}(v_i)=4$ and $v_i$ covers mountains~only~in both drawings, we obtain the open configuration in
\figurename~\ref{fig:openConfigurationA} with 
cost  $2+2\pi$ and  debt $3\pi$. 

If $\deg_{G_i}(v_i)=2$ and $v_i$ covers a pocket in one drawing,
insertion in this drawing costs $\pi$ resulting in total cost $\le 1+2\pi$ or at most $\alpha$ 
if $\pi\le 1/3$. If $\deg_{G_i}(v_i)=2$ and $v_i$ covers only mountains, we have the open configuration in
\figurename~\ref{fig:openConfigurationC} with 
cost  $2+2\pi$ and debt $3\pi$.

\begin{figure}[tb]
  \centering
  \begin{subfigure}[c]{.35\textwidth}
    \centering
    \includegraphics[page=1]{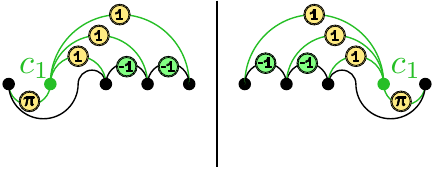}
    \subcaption{$d(\mathcal{C}_a)=3\pi$}  
    \label{fig:openConfigurationA}  
    \end{subfigure}\hfill
    \begin{subfigure}[c]{.275\textwidth}
    \centering
    \includegraphics[page=2]{openConfigurations}
    \subcaption{$d(\mathcal{C}_{b})=3\pi$} 
    \label{fig:openConfigurationB}   
    \end{subfigure}\hfill
    \begin{subfigure}[c]{.275\textwidth}
    \centering
    \includegraphics[page=8]{openConfigurations}
    \subcaption{$d(\mathcal{C}_c)=3\pi$} 
    \label{fig:openConfigurationC}   
    \end{subfigure}
    \vspace{.5\baselineskip}
    
    \begin{subfigure}[b]{.3\textwidth}
    \centering
    \includegraphics[page=4]{openConfigurations}
    \subcaption{$d(\mathcal{C}_d)=\pi$}  
    \label{fig:openConfigurationD}  
    \end{subfigure}\hfill
    \begin{subfigure}[b]{.3\textwidth}
    \centering
    \includegraphics[page=7]{openConfigurations}
    \subcaption{$d(\mathcal{C}_{e})=\pi$} 
    \label{fig:openConfigurationE} 
    \end{subfigure}\hfill
    \begin{subfigure}[b]{.3\textwidth}
    \centering
    \includegraphics[page=6]{openConfigurations}
    \subcaption{$d(\mathcal{C}_f)=\pi$}   
    \label{fig:openConfigurationF} 
    \end{subfigure}
    \vspace{.5\baselineskip}
    
    \begin{subfigure}[b]{.275\textwidth}
    \centering
    \includegraphics[page=3]{openConfigurations}
    \subcaption{$d(\mathcal{C}_g)=2\pi$} 
    \label{fig:openConfigurationG}   
    \end{subfigure}\hfill
    \begin{subfigure}[b]{.275\textwidth}
    \centering
    \includegraphics[page=5]{openConfigurations}
    \subcaption{$d(\mathcal{C}_{h})=2\pi$}  
    \label{fig:openConfigurationH}  
    \end{subfigure}\hfill
    \begin{subfigure}[b]{.35\textwidth}
    \centering    
    \includegraphics[page=9]{openConfigurations}
    \subcaption{$d(\mathcal{C}_\iota)=4\pi$}  
    \label{fig:openConfigurationI}  
    \end{subfigure}
    \vspace{.5\baselineskip}
    
    \begin{subfigure}[b]{.4\textwidth}
	\centering    
    \includegraphics[page=26]{openConfigurations}
    \subcaption{$d(\mathcal{C}_{j})=5\pi$} 
    \label{fig:openConfigurationJ}   
    \end{subfigure}\hfill
    \begin{subfigure}[b]{.4\textwidth}
    \centering
    \includegraphics[page=27]{openConfigurations}
    \subcaption{$d(\mathcal{C}_k)=5\pi$}      
    \label{fig:openConfigurationK}
    \end{subfigure}

  \caption{The set of open configurations. 
    Each subfigure shows the
    forward drawing (left) and the reverse drawing (right) and is captioned by the debt incurred. 
  }
\label{fig:openConfigurationsOverview}
\end{figure}

 It remains to consider  $\deg_{G_i}(v_i)=3$. There are four 
pocket-mountain configurations for two arcs of $G_{i-1}$ covered by $v_i$:
$MM$, $MP$, $PM$, and $PP$ (using $M$ for mountain and $P$ for pocket). Pattern $PP$ costs $1-\pi$, 
pattern $MM$ costs $1+\pi$. Each drawing has
its favorite mixed pattern ($PM$ for forward and $MP$ for reverse) 
with 
cost $0$; the other pattern costs $1$.

There is only one \emph{forward}$|$\emph{reverse} combination, $MM|MM$,
with cost $2+2\pi$ and   debt $3\pi$, leading to the open configuration
in \figurename~\ref{fig:openConfigurationB}. Two combinations,
$MM|PM$ and $MP|MM$, have  cost $2+\pi$ and   debt $2\pi$ resulting in  
open configurations in \figurename~\ref{fig:openConfigurationG} and
\ref{fig:openConfigurationH}, resp. Also, the combinations 
$MM|PP$, $PP|MM$, and $MP|PM$  with costs $2$ and a debt $\pi$  lead to
open configurations in \figurename~\ref{fig:openConfigurationD},
\ref{fig:openConfigurationE}, and \ref{fig:openConfigurationF}, resp.
All other combinations cost at most $\alpha$.

\section{When and how to pay your debts}
\label{sec:openConfig}

In this section, we describe the insertion of  $v_i$ if it covers an 
arc of an open configuration. Note that (1)~every open configuration contains
at least one mountain and at least one pocket in both drawings; (2)~the
largest debt incurred by one open configuration is $5\pi$. 
Open configurations $\mathcal{C}_\iota, \mathcal{C}_j, \mathcal{C}_k$ (with highest debts) are introduced in the discussion below. 

\subparagraph*{Case~1: $\deg_{G_i}(v_i)=2$.} If $v_i$ covers a pocket of an open
configuration $\mathcal{C}$ in either drawing, the insertion costs of
$\pi+(1+\pi)$ cover $d(\mathcal{C})$, as long as $1+2\pi+5\pi\le\alpha$,
that is, $\pi\le 1/8$.
 Assume $v_i$ covers a mountain of open
configuration $\mathcal{C}$ in both drawings; i.e., 
$\mathcal{C}\in\{\mathcal{C}_g,\mathcal{C}_h,\mathcal{C}_\iota\}$.
If $\mathcal{C}\in\{\mathcal{C}_g,\mathcal{C}_h\}$, we obtain the open
configurations in \figurename~\ref{fig:openConfigurationJ} and
\ref{fig:openConfigurationK}, resp., with cost $4+3\pi$ (for both vertices) and
debt $5\pi$. Otherwise $\mathcal{C}=\mathcal{C}_\iota$, and we use the drawings
shown in \figurename~\ref{fig:openConfigurationIplusdeg2:1} (where $v_i$ is
inserted on the left mountain; the other case is symmetric).  The costs are $2+3\pi$
(forward) and $3+2\pi$ (reverse), totaling $5+5\pi\le 3\alpha$, for
$\pi\le 1/8$.


\begin{figure}[tbp]
  \centering   
    \includegraphics[page=28]{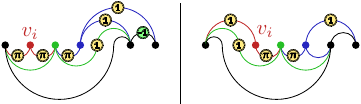}
  \caption{Alternative drawing to handle an open configuration $\mathcal{C}_\iota$ for $\deg_{G_i}(v_i)=2$.}
  \label{fig:openConfigurationIplusdeg2:1}  
\end{figure}

\subparagraph*{Case~2: $\deg_{G_i}(v_i)\ge 5$ and Case~3: $\deg_{G_i}(v_i) \in \{3,4\}$.} 
In the full version, we will discuss both cases in detail while we only mention the main ideas here. Each open configuration includes a mountain that can pay the debt if the configuration is entirely covered. We only focus on the left- and rightmost open configurations \OC{\ell} and \OC{r}. In both cases, we mainly carefully move vertices $v_i$, $c_1$ and $c_2$ of \OC{r} to avoid covered mountains from becoming biarcs, i.e., saving their credits.

\section{Summary \& Conclusions}

\begin{proof}[Proof of Theorem~\ref{thm:1}]
As previously shown, if $\pi \leq 1/8$, we maintain all invariants with $\alpha$ credits per vertex. $G_3$ is a triangle with two pockets on $C_3$ in both orientations, i.e. $G_3$ costs $4\pi$. As $v_1$,$v_2$ and $v_3$ contribute $3\alpha$ credits, there are $6-7\pi > 5$ unused credits after drawing $G_3$. If there remains an open configuration in $G_n$, there is a mountain with a credit paying its debt. Hence, the $5$ unused credits of $G_3$'s drawing remain. As a canonical ordering is computable in $O(n)$ time and we backtrack $O(1)$ steps if needed, the runtime follows.
\end{proof}

We proved the first upper bound of the form $c \cdot n$, with $c < 1$, 
for the total number of monotone biarcs in arc diagrams of
$n$-vertex planar graphs. In our analysis, only three subcases require $\pi \leq
1/8$, i.e., a refinement may provide a better upper bound.  Also, it remains
open if there is a planar graph that requires more biarcs in a monotone arc
diagram than in a general arc diagram. 
Finally, narrowing the gap between lower $\lfloor \frac{n-8}{3} \rfloor$ and upper $\lfloor \frac{15}{16}n - \frac{5}{2}\rfloor$ bounds would be interesting, particularly from the lower bound side.


\bibliographystyle{plainurl}
\bibliography{abbrv,arcs}

\newpage
\appendix

\section*{Appendix}

\section{Proofs of Lemmas~\ref{lem:defaultApproach1} and~\ref{lem:defaultApproach2}}
\label{app:proofsOfLemmas}
\begin{proof}[Proof of Lemma~\ref{lem:defaultApproach1}]
  We place $v_i$ into the rightmost pocket $(p_\ell,p_r)$ it covers and draw all
  edges incident to $v_i$ as proper arcs.
  The path $(p_\ell,v_i,p_r)$ is drawn using pockets, the other new edges are
  drawn as mountains; see \figurename~\ref{fig:naive}. All edges covered to the
  right of $p_r$ (if any) are mountains. To satisfy \ref{i:mountainMoney}, we
  put $1$ credit on every arc between $v_i$ and vertices properly to the right
  of $p_r$. Every such arc $a$ covers one mountain $m_a$ along with the left
  endpoint of $m_a$. By \ref{i:mountainMoney}, $m_a$ has $1$ credit 
  which we move to $a$.
  For the edges between $v_i$ and a vertex properly to the left of $p_\ell$, only
  the leftmost one $(\ell_i,v_i)$ has its left endpoint on
  $C_i$. Therefore we need $\le 1$ additional credit to pay for all mountains
  incident to $v_i$. It remains to consider the pockets. If at most one pocket
  along $(p_\ell,v_i,p_r)$ appears on $C_i$, then we can take the $\pi$ credits
  needed from the now covered pocket $(p_\ell,p_r)$. Otherwise,
  $\deg_{G_i}(v_i)=2$ and we spend only $\pi$ credits overall as there is
  no mountain incident to $v_i$. Thus we pay $\le 1$ credit overall; see~\figurename~\ref{fig:naivePocket2}.

Suppose that $\deg_{G_i}(v_i) \geq 4$. 
If $\ell_i=p_\ell$, then we save a full
credit that we accounted for drawing $(\ell_i,v_i)$ as a mountain. If $v_i$ has at least $3$ neighbors to its left, then $v_i$ covers the left endpoint of an
edge $e$ of $C_{i-1}$ to its left. Hence we save at least $\pi$ credits on
$e$. Otherwise, $v_i$ has exactly $2$ neighbors on its left and at least $2$
 neighbors on its right; see~\figurename~\ref{fig:naivePocket}. In particular, no pocket is incident to
$v_i$ on $C_i$ and we save $\pi$ credits on $(p_\ell,p_r)$. In every case we
save $\ge\pi$ credits and therefore spend $\le 1-\pi$ credits to insert $v_i$. 
\end{proof} 
\begin{proof}[Proof of Lemma~\ref{lem:defaultApproach2}]
  We push down the leftmost mountain $(m_\ell,m_r)$ and place $v_i$ above it;
  see \figurename~\ref{fig:naiveMountains}. This operations transforms each
  mountain of $G_{i-1}$ with left endpoint $m_\ell$ into a
  down-up biarc. Costs for such biarcs are covered by their corresponding
  mountain credits due to \ref{i:mountainMoney}. 

  The insertion of $v_i$ creates a new pocket $(m_\ell,v_i)$ and a mountain for every
  neighbor in $(m_r,\ldots,r_i)$. We put $\pi$ credits on the pocket to satisfy
  \ref{i:pocketMoney}. On each mountain created, we put $1$ credit to satisfy
  \ref{i:mountainMoney}. Since $v_i$ covers only mountains, every mountain
  incident to $v_i$, except for $(v_i,m_r)$, covers the left endpoint of a
  mountain from $G_{i-1}$. Hence from each of these covered mountains we can
  steal $1$ credit; see \figurename~\ref{fig:naiveMountains}. Overall, $1 + \pi$
  credits suffice.

  If $\deg_{G_i}(v_i) \geq 4$, we can push down the rightmost covered mountain
  instead. This way we create exactly two new mountains whose left endpoint is
  on $C_i$. We also cover the left endpoint of $\deg_{G_i}(v_i)-3$ mountains from
  $C_{i-1}$. Hence, we can steal the credits from these mountains. The resulting
  total cost is $2-(\deg_{G_i}(v_i)-3)=5-\deg_{G_i}(v_i)$; see
  \figurename~\ref{fig:degree4+Mountains}.
\end{proof}

\section{Omitted Cases of Section~\ref{sec:openConfig}}
\label{ap:deg-3-4}

\subparagraph*{Case~2: $\deg_{G_i}(v_i)\ge 5$.} Here, $v_i$ may cover many open
configurations completely or partially (i.e. not
all arcs of the configuration are covered). Intuitively, fully covered open
configurations pay for themselves, only at most two partially covered
ones need some work.  Denote by $\mathcal{C}_\ell$ the open configuration
incident to the leftmost arc covered by $v_i$,  incident to vertex
$\ell_i$. If this arc is not part of an open configuration, we write
$\mathcal{C}_\ell=\emptyset$ or that $\mathcal{C}_\ell$ does not
exist. Analogously, $\mathcal{C}_r$ denotes the open configuration incident to the
rightmost covered arc, incident to vertex $r_i$. $\mathcal{C}_\ell$ and $\mathcal{C}_r$---if
existent---can be fully or partially covered by $v_i$.

Reonsider how to insert $v_i$
(cf.~Lemmas~\ref{lem:defaultApproach1}~and~\ref{lem:defaultApproach2}): We
show how to insert $v_i$ in the forward drawing at a cost
of $\le 1-\pi$, while also paying $d(\mathcal{C}_r)$ and debts of
all fully covered open configurations except $\mathcal{C}_\ell$. By symmetry, $v_i$ can be  inserted into the reverse drawing at a
cost of $\le 1-\pi$, including payment for $d(\mathcal{C}_\ell)$, achieving a total cost of at most $2-2\pi < \alpha$.

\subparagraph*{Case 2a: $v_i$ covers mountains only.} By (1) $v_i$ does not fully
cover any open configuration. By Lemma~\ref{lem:defaultApproach2} inserting
$v_i$ costs $\le 0$ credits, which together with $d(\mathcal{C}_r)\le 5\pi$
yields an insertion cost of $\le 1-\pi$, as long as $\pi\le 1/6$.

\subparagraph*{Case 2b: $v_i$ covers a pocket $p$ incident to $r_i$.} We insert $v_i$ into $p$  paying $1$ for edge $(\ell_i,v_i)$ and gaining $1$ for each
covered mountain not incident to $\ell_i$ and $\pi$ for each covered pocket. Worst case, all covered arcs to the left are pockets---except for the one
incident to $\ell_i$, which is a mountain, i.e., insertion costs are
$\le 1-(\deg_{G_i}(v_i)-3)\pi\le 1-2\pi$~(*). For each fully covered open
configuration $\mathcal{C}$, $v_i$ covers a mountain $M(\mathcal{C})$. If
$M(\mathcal{C})$ is incident to $\ell_i$, $\mathcal{C}=\mathcal{C}_\ell$
and we account for it in the reverse drawing. Otherwise, we gain at least $1$ credit for covering
$M(\mathcal{C})$, i.e., we gain $1-\pi$ credits compared to (*). By (2) this
settles $d(\mathcal{C})$ for $\pi\le 1/6$.

The case where $\mathcal{C}_r$ is only partially covered remains. By (*), w.l.o.g. assume  $d(\mathcal{C}_r)>\pi$. There are three subcases, depending on
the edge $e$ on $C_{i-1}$ to the left of $p$ 
covered by $v_i$.

\emph{Case 2b.1}: $e$ is a mountain that is not part of an open configuration.
Then again we gain $1-\pi$ credits by covering it, which is enough to pay
$d(\mathcal{C}_r)$ if $\pi\le 1/6$.

\emph{Case 2b.2}: $e$ is a mountain that is part of open configuration
$\mathcal{C}_M\ne\mathcal{C}_r$. As $v_i$ covers the left endpoint $q$ of $e$,
we can steal $1$ credit from \emph{each} mountain incident to $q$. Note
that for all open configurations but $\mathcal{C}_c$, $\mathcal{C}_e$,
$\mathcal{C}_\iota$, $\mathcal{C}_j$ there are at least two mountain edges
incident to $q$. In these cases we save an extra credit paying for
$d(\mathcal{C}_r)$. Only cases
$\mathcal{C}_M\in\{\mathcal{C}_c,\mathcal{C}_e,\mathcal{C}_\iota,\mathcal{C}_j\}$~remain.

If $\mathcal{C}_M=\mathcal{C}_c$, $v_i$ also covers the left endpoint of
$\mathcal{C}_M$. We move vertex $c_1\in\mathcal{C}_M$ into $p$ so that the
black edge underneath can be drawn as a mountain instead of as a biarc; saving another $1-\pi$ credits to pay $d(\mathcal{C}_r)$ if $\pi\le 1/4$.
If $\mathcal{C}_M=\mathcal{C}_e$, we redraw $\mathcal{C}_M$ so that
vertex $c_1$ is put into the right pocket instead. Costs and the deficit are
the same as for the original drawing. Now $e$ is a pocket and we proceed with
Case~2b.4 below. 
    \begin{figure}[tb]
    \centering
    \includegraphics[page=36]{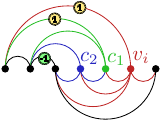}
  \caption{Alternative drawing to handle open configuration $\mathcal{C}_\iota$ if $\deg_{G_i}(v_i)\ge 5$.}
  \label{fig:openConfigurationIplusdeg2:2}
\end{figure}
If $\mathcal{C}_M=\mathcal{C}_\iota$, we save many credits by redrawing
$\mathcal{C}_M$ to move both $c_1$ and $c_2$ into $p$ so that $c_2,c_1,v_i$
appear in this left-to-right-order (see
\figurename~\ref{fig:openConfigurationIplusdeg2:2}).
Otherwise, $\mathcal{C}_M=\mathcal{C}_j$ and we move vertex
$c_2\in\mathcal{C}_M$ into $p$ s.t. both green edges underneath become
mountains instead of biarcs. As $v_i$ covers $c_1\in\mathcal{C}_M$, we
steal the credit of all three mountains incident to $c_1$, resulting in a huge
gain.

\emph{Case 2b.3}: $e$ is a mountain and part of $\mathcal{C}_r$, i.e.,
$\mathcal{C}_r=\mathcal{C}_\iota$. $v_i$ covers the left endpoint of the left
mountain of $\mathcal{C}_r$. We steal the credit from the black mountain
underneath to pay $d(\mathcal{C}_r)$.

\emph{Case 2b.4:} $e$ is a pocket (possibly of an open configuration). Then,
$\mathcal{C}_r\in\{\mathcal{C}_a, \mathcal{C}_b, \mathcal{C}_c, \mathcal{C}_g,
\mathcal{C}_j, \allowbreak \mathcal{C}_k\}$. If
$\mathcal{C}_r\in\{\mathcal{C}_j,\mathcal{C}_k\}$, then $\mathcal{C}_r$ has two
pockets and $p$ can be either one. If $p$ is the left pocket, $e$ is
not part of $\mathcal{C}_r$. We move vertices $c_1,c_2\in\mathcal{C}_r$ into
$e$ to the right of $v_i$. This allows to redraw the black edge underneath
$c_1$ and $c_2$ as a mountain rather than as a biarc (see
\figurename~\ref{fig:payfordcr:1}--\ref{fig:payfordcr:2}). The left endpoint
of this new mountain is covered by $v_i$ and we gain $1-\pi$ credits paying
$d(\mathcal{C}_r)$ if $\pi\le 1/6$. An analogous rearrangement handles cases
$\mathcal{C}_r\in\{\mathcal{C}_a, \mathcal{C}_b, \mathcal{C}_c,
\mathcal{C}_g\}$. It remains to address the case where $p$ is the right pocket
of $\mathcal{C}_r\in\{\mathcal{C}_j,\mathcal{C}_k\}$.

\begin{figure}[tb]
  \centering
   \begin{subfigure}[b]{.24\textwidth}
    \centering
    \includegraphics[page=30]{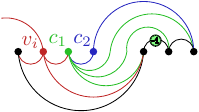}
	\subcaption{$\mathcal{C}_r=\mathcal{C}_j$}     
    \label{fig:payfordcr:1}   
    \end{subfigure}\hfil
    \begin{subfigure}[b]{.225\textwidth}
    \centering
    \includegraphics[page=31]{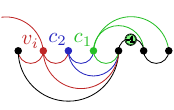}
    \subcaption{$\mathcal{C}_r=\mathcal{C}_k$}    
    \label{fig:payfordcr:2}    
    \end{subfigure}\hfil
    \begin{subfigure}[b]{.225\textwidth}
    \centering
    \includegraphics[page=33]{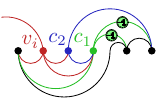}
	\subcaption{$\mathcal{C}_r=\mathcal{C}_j$}       
    \label{fig:payfordcr:3} 
    \end{subfigure}\hfil
    \begin{subfigure}[b]{.21\textwidth}
    \centering
    \includegraphics[page=32]{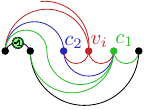}
	\subcaption{$\mathcal{C}_r=\mathcal{C}_k$}       
    \label{fig:payfordcr:4} 
    \end{subfigure}
  \caption{How to pay $d(\mathcal{C}_r)$ if $v_i$ is inserted into a pocket of
    $\mathcal{C}_j$ and $\mathcal{C}_k$.}
  \label{fig:payfordcr}
\end{figure}

If $\mathcal{C}_r=\mathcal{C}_j$, we move both $c_2$ and $v_i$ to the left
pocket of $\mathcal{C}_r$, thereby allowing both edges from $c_1$ to the right
to be drawn as mountains (see \figurename~\ref{fig:payfordcr:3}). As $c_1$ is
covered by $v_i$, we can steal the credits from these two mountains, which is
more than enough to pay $d(\mathcal{C}_r)$.

For $\mathcal{C}_r=\mathcal{C}_k$, consider edge $d$ on $C_{i-1}$ to the
left of $e$.  $d$ and its left endpoint are covered by $v_i$. If $d$ is
a pocket (of an open configuration or not), we redraw $\mathcal{C}_r$ by
moving $c_1$ to the pocket underneath $\mathcal{C}_r$ and $c_2$ to $d$. Then we
insert $v_i$ into $d$ to the left of $c_2$ (see
\figurename~\ref{fig:payfordcr:4}). Compared to the original drawing, we gain
two credits to pay $d(\mathcal{C}_r)$.  Otherwise, $d$ is a mountain. If
$d$ is part of $\mathcal{C}_\ell$ or $d$ is not part of an open configuration,
we gain $1-\pi$ credits by covering its left endpoint, which pays
$d(\mathcal{C}_r)$ if $\pi\le 1/6$. Next, suppose that $d$ is part of open
configuration $\mathcal{C}'\ne\mathcal{C}_\ell$. Then
$\mathcal{C}'\in\{\mathcal{C}_a, \mathcal{C}_b, \mathcal{C}_c, \mathcal{C}_d,
\mathcal{C}_e, \mathcal{C}_g, \mathcal{C}_\iota, \mathcal{C}_j, \mathcal{C}_k\}$ and
$\mathcal{C}'$ and all its vertices are fully covered by $v_i$. For
$\mathcal{C}'\in\{\mathcal{C}_a, \mathcal{C}_b, \mathcal{C}_d, \mathcal{C}_g,
\mathcal{C}_\iota, \mathcal{C}_k\}$ there is at least one more mountain below a
mountain of $\mathcal{C}'$ from which we steal its credit to pay
$d(\mathcal{C}_r)$ if $\pi\le 1/5$. For $\mathcal{C}'=\mathcal{C}_e$, redraw
$\mathcal{C}'$ by moving $c_1$ into the right pocket underneath. The costs are
unchanged, $d$ is a pocket now and we continue as described above.
For $\mathcal{C}'=\mathcal{C}_c$ we move $c_1$ to $e$ so that the black biarc
underneath can be redrawn as a mountain. As a result, we gain $1-\pi$ credits to
pay $d(\mathcal{C}_r)$. Case
$\mathcal{C}'=\mathcal{C}_j$ can be resolved symmetrically, gaining $2-\pi$ credits to pay
$d(\mathcal{C}_r)$.

\subparagraph*{Case 2c: $v_i$ covers at least one pocket but no pocket incident to
  $r_i$.} We insert $v_i$  into the rightmost covered pocket $p$ which
covers at least one mountain to the right of $p$. All edges
from $v_i$ (strictly) to the right of $p$ are drawn as mountains and we put $1$
credit on them stolen from the mountain underneath (cf.
Lemma~\ref{lem:defaultApproach1}). We proceed as in Case~2b to the left of $v_i$, where the same analysis holds for the covered open configurations.

By~(1) no fully covered open configuration lies strictly to the right of
$p$. However, $p$ may be part of an open configuration $\mathcal{C}$ extending to the right of $p$. In such a case, either
$\mathcal{C}=\mathcal{C}_\ell$ or $\mathcal{C}$ is fully covered by $v_i$
(because there is no open configuration with a pocket-mountain subsequence in
its profile that does not end with this mountain). If
$\mathcal{C}=\mathcal{C}_\ell$, then we do not account for it here but in the
reverse drawing. So  assume that $\mathcal{C}$ is fully covered by
$v_i$.

If
$\mathcal{C}\in\{\mathcal{C}_a, \mathcal{C}_b, \mathcal{C}_d, \mathcal{C}_g,
\mathcal{C}_k\}$, there is at least one more mountain underneath the
mountain of $\mathcal{C}$ whose left endpoint is also covered by $v_i$. We steal the
credit of this mountain  to pay $d(\mathcal{C})$. If
$\mathcal{C}=\mathcal{C}_\iota$,   there is an extra mountain
underneath the left mountain of $\mathcal{C}$ whose left endpoint is covered by
$v_i$. The extra credit pays $d(\mathcal{C})$. For
$\mathcal{C}=\mathcal{C}_j$, we use the same rearrangement as in
\figurename~\ref{fig:payfordcr:3} to gain a large surplus, except $v_i$ also covers $c_2$. Only $\mathcal{C}\in\{\mathcal{C}_c, \mathcal{C}_e\}$ remains. 

If $\mathcal{C}=\mathcal{C}_e$, we redraw $\mathcal{C}$ by putting $c_1$
along with $v_i$ into the right pocket of $\mathcal{C}$. The costs and deficit
are the same as for the original drawing, however, we save the credit that
used to be on the mountain from $v_i$ to the rightmost vertex of
$\mathcal{C}$. We use this credit to pay $d(\mathcal{C})$. 

Finally, suppose $\mathcal{C}=\mathcal{C}_c$. If $\mathcal{C}=\mathcal{C}_r$,
$v_i$ has at least two edges to vertices strictly to the left of $p$ and we
can argue exactly as in Case~2b  by considering the edge $e$ to the left of
$p$. Otherwise, $v_i$ covers the mountain $m$ immediately to the right of
$p$. If $m$ is the leftmost edge of an open configuration, this
configuration is
$\mathcal{C}_r\in\{\mathcal{C}_f,\mathcal{C}_h,\mathcal{C}_\iota\}$ and
there is another mountain $m'$ underneath $m$. As $v_i$ covers the left endpoint
of both $m$ and $m'$, we steal the credit of $m'$ to pay
$d(\mathcal{C})+d(\mathcal{C}_r)\le 3\pi+4\pi$, as long as $\pi\le
1/7$. Otherwise, $m$ is not part of an open configuration and we push down $m'$
to insert both $c_1\in\mathcal{C}$ and $v_i$ there. As a result, we draw the
edge underneath $c_1$ in $\mathcal{C}$ in the original drawing (see
\figurename~\ref{fig:deg5Cc:1}) as a mountain rather than as a biarc, gaining
$1$ credit. Other costs remain the same as drawing $m'$ as a biarc is
compensated by having one edge less from $v_i$ to vertices on the right (see
\figurename~\ref{fig:deg5Cc:2}).

\begin{figure}[tb]
  \centering
  \begin{subfigure}[b]{.4\textwidth}
    \centering
    \includegraphics[page=34]{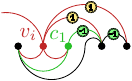}
	\subcaption{before}     
    \label{fig:deg5Cc:1}   
    \end{subfigure}\hfil
    \begin{subfigure}[b]{.4\textwidth}
    \centering
    \includegraphics[page=35]{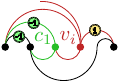}
    \subcaption{after}  
    \label{fig:deg5Cc:2}  
    \end{subfigure}
  \caption{Redrawing if $p$ is part of $\mathcal{C}=\mathcal{C}_c$.}
  \label{fig:deg5Cc}
\end{figure}

Regardless of whether or not $p$ is part of an open configuration, if
$\mathcal{C}_r$ exists we still need to pay its debt. The case that $p$ is part
of $\mathcal{C}_r$ has been discussed above. Hence suppose that $\mathcal{C}_r$
exists but $p$ is not part of it. Recall that by~(1) no fully covered open
configuration lies strictly to the right of $p$. Hence $\mathcal{C}_r$ is only
partially covered by $v_i$ and has a mountain $m$ as its leftmost edge. Therefore
$\mathcal{C}_r\in\{\mathcal{C}_f, \mathcal{C}_h, \mathcal{C}_\iota\}$ where
$m$ has another mountain $m'$ underneath. As $v_i$ covers the left
endpoint of $m$ and $m'$, we gain an extra credit from $m'$ to pay $d(\mathcal{C}_r)$.

\subparagraph*{Case 3: $\deg_{G_i}(v_i) \in \{3,4\}$}
Cases $\deg_{G_i}(v_i) = 3$ and $\deg_{G_i}(v_i) = 4$ are treated together as in many subcases similar arguments can be used. 
Drawings referring to both cases illustrate Case $\deg_{G_i}(v_i) = 4$ as removing the leftmost vertex yields the drawing for $\deg_{G_i}(v_i) = 3$.

\subparagraph*{Case~3a: $v_i$ covers at least the two leftmost edges of configuration \OC{r}.} 
Note that, the case where $v_i$ covers the two rightmost edges of \OC{\ell} follows symmetrically. 

\emph{Case 3a.1:} $\deg_{G_i}(v_i)=3$. 
Here, $v_i$ covers a pocket in both orientations and is inserted above a $PM$ configuration in forward  or above an $MP$ configuration in reverse orientation unless $\mathcal{C}_r = \mathcal{C}_f$. 
In this orientation, inserting $v_i$ costs $0$ credits  while in the remaining orientation, $v_i$ costs at most $1$ by Lemma~\ref{lem:defaultApproach1}. 
Hence, $v_i$ can pay the debt of the (partially) covered open configuration  yielding a total cost of at most $1+5\pi$ for $v_i$, i.e., at most $\alpha$ for $\pi \le 1/6$. 
If $\mathcal{C}_r = \mathcal{C}_f$, we push down the mountain covered by $c_1$ and place $c_1$ and $v_i$ above the new biarc; see Fig~\ref{fig:case3a1}. As a result, we do not have to charge any edge incident to $c_1$ and only have to pay for a pocket and a mountain incident to $v_i$. As a pocket is covered, the forward orientation costs $1$. Since the reverse orientation is symmetric, the total cost is $2$.

\begin{figure}[b]
  \centering
  {\includegraphics[page=22]{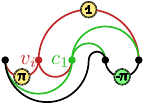}}
  \caption{Non-default subcase of Case~3a.1.}
  \label{fig:case3a1}
\end{figure}

\emph{Case 3a.2:} $\deg_{G_i}(v_i) = 4$ and $v_i$ completely covers \OC{r}. 
Consider the forward orientation.  

If \OC{r} $\in \{\mathcal{C}_a,\mathcal{C}_b,\mathcal{C}_c,\mathcal{C}_d,\mathcal{C}_g\}$, $v_i$ covers vertex $c_1$ of \OC{r}. 
We redraw \OC{r} placing both $v_i$ and $c_1$ above $e_\ell$ (if required we push down $e_\ell$); see Fig.~\ref{fig:case3a2-1}. 
While this creates a biarc incident to $v_i$, all mountains incident to $c_1$ are covered in the final drawing and hence do not need to be charged. Moreover, placing $c_1$ above $e_\ell$ allows to release the credit on the leftmost mountain covered by $c_1$. This results in a cost of at most $1+\pi$ credits in the forward orientation for both vertices. In the reverse orientation, $c_1$ can cost at most $1+\pi$ but creates a pocket for $v_i$ which then by Lemma~\ref{lem:defaultApproach1} will cost at most $1-\pi$. By combining these costs with the debt of \OC{\ell}, inserting $v_i$ and $c_1$ will cost at most $3+6\pi \le 2\alpha$ when $\pi \le 1/8$.

\begin{figure}[t]
 \centering%
 \begin{subfigure}[b]{.3\textwidth}
    \centering
    \includegraphics[page=1,scale=1]{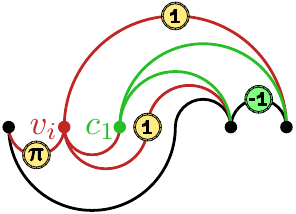}
	\subcaption{}       
    \label{fig:case3a2-1} 
    \end{subfigure}\hfil
    \begin{subfigure}[b]{.3\textwidth}
    \centering
    \includegraphics[page=2,scale=1]{newDegree3And4figures}
    \subcaption{}  
    \label{fig:case3a2-2}  
    \end{subfigure}\hfil
    \begin{subfigure}[b]{.3\textwidth}
    \centering
    \includegraphics[page=3,scale=1]{newDegree3And4figures}
    \subcaption{}
    \label{fig:case3a2-3}    
    \end{subfigure}
  \caption{Figures for Case~3a.2.}
  \label{fig:case3a2}
\end{figure}

Suppose $\mathcal{C}_r \in \{\mathcal{C}_f,\mathcal{C}_h, \mathcal{C}_e\}$. 
If $\mathcal{C}_r \in \{\mathcal{C}_f,\mathcal{C}_h\}$, $v_i$ is placed into pocket $e_r$ costing at most $1-\pi$ by Lemma~\ref{lem:defaultApproach1}. 
If $\mathcal{C}_r = \mathcal{C}_e$, we achieve the same cost for $v_i$ by putting $c_1$ of \OC{r} in its right covered pocket instead of its left one. 
In each of $\{\mathcal{C}_f,\mathcal{C}_h, \mathcal{C}_e\}$, observe that the left endpoint of \OC{r} is covered without creating a biarc incident to $c_1$, i.e., the mountain incident to $c_1$ is covered and does not need a credit. 
Hence, the total cost for the forward orientation drawing here is at most $1-\pi$. 
In the reverse orientation, $c_1$ can cost at most $1+\pi$ but creates a pocket for $v_i$ which then by Lemma~\ref{lem:defaultApproach1} will cost at most $1-\pi$. 
Therefore, together with the debt of \OC{\ell}, inserting $v_i$ and $c_1$ here can cost at most $3+4\pi \le 2\alpha$ as long as $\pi \le 1/6$.

It remains to consider $\mathcal{C}_r \in \{\mathcal{C}_\iota,\mathcal{C}_j, \mathcal{C}_k\}$. 
For $\mathcal{C}_r \in \{\mathcal{C}_\iota,\mathcal{C}_j\}$, the default drawing is used for $c_1$ of $\mathcal{C}_r$ and then $v_i$ and $c_2$ of $\mathcal{C}_r$ are put on $e_\ell$ which is a pocket; see Figs.~\ref{fig:case3a2-2} and~\ref{fig:case3a2-3}. 
Since $v_i$ covers both $c_1$ and $c_2$ and no biarc is created, mountains incident to $c_1$ and $c_2$ do not need to carry a credit. 
Further, as $v_i$ and $c_2$ are in the same pocket, $v_i$ is only incident to one mountain that steals its credit from the mountain below \OC{r}. 
As a result, the forward orientation costs at most $\pi$ for all three vertices. 
The reverse orientation costs at most $3+3\pi$ by Lemmas~\ref{lem:defaultApproach1} and~\ref{lem:defaultApproach2}, and the cost for all three vertices  is at most $3+4\pi \le 3\alpha$ as long as $\pi \le 3/7$. 
If $\mathcal{C}_r = \mathcal{C}_k$, a symmetric argument w.r.t. $\mathcal{C}_r = \mathcal{C}_j$ applies in the reverse orientation.

\begin{figure}[tb]
 \centering%
 \begin{subfigure}[b]{.3\textwidth}
    \centering
    \includegraphics[page=4,scale=1]{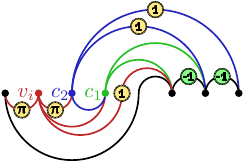}
    \subcaption{} 
    \label{fig:case3a3-1}   
    \end{subfigure}\hfil
    \begin{subfigure}[b]{.3\textwidth}
    \centering
    \includegraphics[page=5,scale=1]{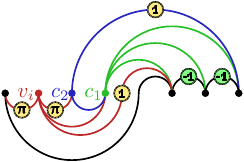}
    \subcaption{} 
    \label{fig:case3a3-2}   
    \end{subfigure}\hfil
    \begin{subfigure}[b]{.3\textwidth}
    \centering
    \includegraphics[page=6,scale=1]{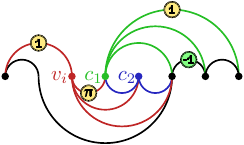}
    \subcaption{} 
    \label{fig:case3a3-3}  
    \end{subfigure}
  \caption{Figures for Case~3a.3.}
  \label{fig:case3a3}
\end{figure}

\emph{Case 3a.3:} $\deg_{G_i}(v_i) = 4$,  $\mathcal{C}_r \in \{\mathcal{C}_\iota,\mathcal{C}_j,\mathcal{C}_k\}$  and $v_i$ covers the two leftmost edges of~$\mathcal{C}_r$.  
If $\mathcal{C}_r \in \{\mathcal{C}_\iota,\mathcal{C}_j\}$, in the forward orientation $v_i$ and $c_1$ and $c_2$ of configuration \OC{r} are placed above edge $e_\ell$ which is pushed down if needed; see Figs.~\ref{fig:case3a3-1} and~\ref{fig:case3a3-2}. 
Observe that $v_i$ is incident to a biarc and to no mountains, whereas $c_1$ of \OC{r} is covered and its incident mountains do not need to carry credits. 
Finally, $c_2$ is incident to at most two mountains. 
As two mountains below \OC{r} are covered and two pockets incident to $v_i$ are created, a total cost of $1+2\pi$ suffices for the forward orientation drawing of all three vertices. 
By Lemmas~\ref{lem:defaultApproach1} and~\ref{lem:defaultApproach2}, the reverse orientation cannot cost more than $3+3\pi$. Hence, the total cost is at most $4+ 10\pi$ for all three vertices and $d(\mathcal{C}_\ell) \leq 5\pi$. 
If  $\mathcal{C}_r = \mathcal{C}_k$, all of $v_i$ and $c_1$ and $c_2$ of \OC{r} are placed in the reverse orientation above the edge $e_\ell$ which we push down as needed; see Fig.~\ref{fig:case3a3-3}. 
As a result, we pay for one mountain incident to $c_1$ and one mountain and one valley incident to $v_i$ while also covering the left mountain of \OC{r} without creating a biarc. 
Hence, the reverse orientation of all three vertices costs $1 + \pi$. 
As the forward orientation costs at most $3+3\pi$ and \OC{\ell}'s debt is at most $5\pi$, all three vertices are inserted for $4+9\pi \le 3\alpha$ as long as $\pi \le 1/6$.

\subparagraph*{Case~3b: $v_i$ covers the leftmost edge of configuration $\mathcal{C}_r$.}  
This case is symmetric to $v_i$ covering the rightmost edge of configuration \OC{\ell}. Note
 that if $\deg_{G_i}(v_i) = 4$ and $e_m$ belongs to open configuration \OC{\ell}, one of the subcases of Case 3a applies. 
Let $e^\ast$ be the covered edge incident to the left endpoint of $\mathcal{C}_r$ from the left, i.e., if $\deg_{G_i}(v_i) = 3$, $e^\ast = e_\ell$, else $e^\ast = e_m$.

\begin{figure}[tb]
 \centering%
 \begin{subfigure}[b]{.3\textwidth}
    \centering
    \includegraphics[page=7,scale=1]{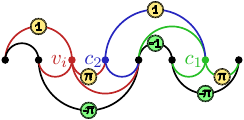}
    \subcaption{}  
    \label{fig:case3b1-1}  
    \end{subfigure}\hfil
    \begin{subfigure}[b]{.3\textwidth}
    \centering
    \includegraphics[page=8,scale=1]{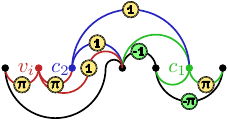}
    \subcaption{}  
    \label{fig:case3b1-2}
    \end{subfigure}\hfil
    \begin{subfigure}[b]{.3\textwidth}
    \centering
    \includegraphics[page=9,scale=1]{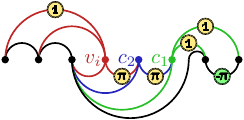}
    \subcaption{} 
    \label{fig:case3b1-3}   
    \end{subfigure}
  \caption{Figures for the forward orientation of Case~3b.1.}
  \label{fig:case3b1}
\end{figure}

\emph{Case 3b.1:} $\mathcal{C}_r = \mathcal{C}_k$. 
Consider two subcases for $e^\ast$ to bound the cost of the forward orientation.
If $e^\ast$ is a pocket, both $v_i$ and $c_2$ are put above $e^\ast$ and $c_1$ is put into the pocket covered by \OC{r}; see Fig.~\ref{fig:case3b1-1}. 
No biarcs are created and it suffices to pay for mountain $(c_1,c_2)$, two pockets, and pocket or mountain $(v_i,\ell_i)$. Also two pockets and a mountain which no longer needs its credit are covered. 
Hence, at most $1$ credit is needed for all of $v_i$, $c_1$ and $c_2$. 
If $e^\ast$ is a mountain, we distinguish two subcases based on $\deg_{G_i}(v_i)$. 
If $\deg_{G_i}(v_i) = 3$, $e^\ast$ is pushed down and $v_i, c_1,$ and $c_2$ are placed as in the pocket case; see Fig.~\ref{fig:case3b1-2}. 
Since $\deg_{G_i}(v_i) = 3$, $(v_i,\ell_i)$ is a pocket. 
This allows to pay for the additional mountain incident to $c_2$ and the biarc incident to $v_i$ and still have a drawing of cost $2+2\pi$ for all three vertices. 
If $\deg_{G_i}(v_i) = 4$, instead the mountain covered by $\mathcal{C}_r$ is pushed down and all of $v_i$, $c_1$ and $c_2$ are placed there; see Fig.~\ref{fig:case3b1-3}. 
We pay for mountain $(v_i,\ell_i)$, the two pockets incident to $c_2$ and the two mountains incident to $c_1$. 
However, we also cover $e^\ast$ yielding a drawing of cost $2+\pi$ for all three vertices. 
Thus, we pay at most $2+2\pi$ for the forward orientation. 

\begin{figure}[tb]
 \centering%
  \begin{subfigure}[b]{.4\textwidth}
    \centering
    \includegraphics[page=10,scale=1]{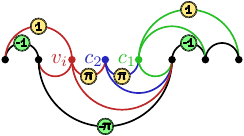}
    \subcaption{} 
    \label{fig:case3b1r-1}   
    \end{subfigure}\hfil
    \begin{subfigure}[b]{.4\textwidth}
    \centering
    \includegraphics[page=11,scale=1]{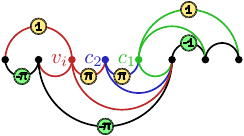}
    \subcaption{} 
    \label{fig:case3b1r-2}
    \end{subfigure}
    
    \begin{subfigure}[b]{.4\textwidth}
    \centering
    \includegraphics[page=12,scale=0.9]{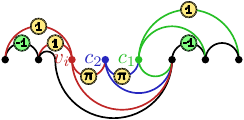}
    \subcaption{} 
    \label{fig:case3b1r-3}
    \end{subfigure}\hfil
    \begin{subfigure}[b]{.4\textwidth}
    \centering
    \includegraphics[page=13,scale=0.9]{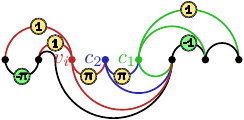}
    \subcaption{} 
    \label{fig:case3b1r-4}
    \end{subfigure}
  
  \caption{Figures for the reverse orientation of Case~3b.1.}
  \label{fig:case3b1r}
\end{figure}

In the reverse orientation, we always place all of $v_i$, $c_1$ and $c_2$ above $e^\ast$; see Fig.~\ref{fig:case3b1r}. 
Then, for $c_1$, we only have to charge the mountain to its rightmost neighbor which we can pay with the credit on the left mountain covered by $c_1$. 
Additionally, we have to pay for the pockets $(c_1,c_2)$ and $(v_i,c_2)$. 
Finally, we may have to pay for up to two mountains from $v_i$ to its left neighbors. 
Thus, it suffices to pay at most $2+2\pi$ for the reverse orientation.

Summing up both orientations, for all three vertices and $d(\mathcal{C}_\ell) \leq 5\pi$, the total cost is most $4+11\pi$ which is at most $3\alpha$ as long as $\pi \le 1/7$.

\emph{Case 3b.2:} $\mathcal{C}_r = \mathcal{C}_j$. 
Note that vertex $c_2$ of \OC{r} is not incident to $v_i$. Since $\mathcal{C}_r$ is still an open configuration, $c_2$ did not get new neighbors after being inserted. 
Hence, we can change the canonical ordering inserting $v_i$ before $c_2$. Then $\mathcal{C}_r = \mathcal{C}_g$ which is discussed later in this section. Open configurations \OC{j} and \OC{k} have been discussed and in the following $d(\mathcal{C}_\ell) \leq 4\pi$.

\emph{Case 3b.3:}  $\mathcal{C}_r = \mathcal{C}_\iota$. 
If $\deg_{G_i}(v_i)=4$, we use the same trick as in Case 3b.2 and insert $c_2$ after $v_i$, i.e., $\mathcal{C}_r = \mathcal{C}_b$ -- see Case 3b.4. 
However, if $\deg_{G_i}(v_i)=3$, we cannot do this since $v_i$ is symmetric to $c_2$ w.r.t. $c_1$. 
Instead, $v_i$ is always placed above $e_\ell$ by using the default drawing. 
We now discuss the costs of this placement depending on the state of $e_\ell$.

If $e_\ell$ is a pocket, placing $v_i$ will cost $0$ credits while  the credits of the mountain covered by \OC{r} whose left endpoint gets covered can be reclaimed. 
As a result, the forward orientation costs $-1$ credit. Since the reverse orientation costs at most $1+\pi$ and the debts of the two open configurations are at most $8\pi$, the total cost of $v_i$ is at most $9\pi \le \alpha$ as long as $\pi \le 1/4$. 

If $e_\ell$ is a mountain, in the forward drawing, we pay for one mountain and one pocket while $v_i$ covers the left endpoint of the mountain covered by \OC{r} achieving a cost of $\pi$. 
If $e_\ell$ is a mountain in reverse orientation, $\mathcal{C}_\ell \in \{\emptyset,\mathcal{C}_g,\mathcal{C}_\iota\}$ and we use the default drawing for $v_i$ costing $1+\pi$. 
If $\mathcal{C}_\ell = \emptyset$, then the total cost for inserting $v_i$ is $1+6\pi \leq \alpha$ as long as $\pi \le 1/7$. 
If $\mathcal{C}_\ell \in \{\mathcal{C}_g,\mathcal{C}_\iota\}$, by symmetry, the cost is again reduced to $\pi$ and achieving a total cost of at most $10\pi$. 
If $e_\ell$ is a pocket in reverse orientation, $v_i$ and $c_1$ are put in $e_\ell$ while we push down the rightmost edge covered to put $c_2$ there; see Fig.~\ref{fig:case3b3}. 
Then, we have to pay for the two mountains incident to $c_2$ and the two pockets incident to $v_i$ but we cover a pocket and a mountain completely. Hence, the total cost for inserting all three vertices is $1+2\pi$. Since the forward orientation costs at most $3+3\pi$  and $d(\mathcal{C}_\ell) \leq 4\pi$, all three vertices cost at most $4+9\pi \le 3\alpha$ as long as $\pi \le 1/6$. Since \OC{\iota}, \OC{j} and \OC{k} are handled, in the following $d(\mathcal{C}_\ell) \leq 3\pi$.

\begin{figure}[tb]
  \centering
  {\includegraphics[page=24]{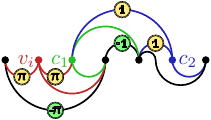}}
  \caption{Non-default subcase of Case~3b.3.}
  \label{fig:case3b3}
\end{figure}

\begin{figure}[tb]
  \centering
  {\includegraphics[page=15]{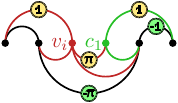}}
  \caption{Forward orientation drawing for Case~3b.4.}
  \label{fig:case3b4}
\end{figure}

\emph{Case 3b.4:} $\mathcal{C}_r \in \{\mathcal{C}_a,\mathcal{C}_b,\mathcal{C}_c,\mathcal{C}_d,\mathcal{C}_g\}$. As in previous cases, consider $e^\ast$ and $\deg_{G_i}(v_i)$.

If $e^\ast$ is a pocket, we put both $v_i$ and $c_1$ in $e^\ast$; see Fig.~\ref{fig:case3b4}, paying for pocket $(v_i,c_1)$ (which can be done using $e^\ast$'s $\pi$ credits) and edge $(v_i,\ell_i)$ which may be pocket or mountain depending on the degree of $v_i$. 
The mountains incident to $c_1$ can be paid by the credits of mountains covered by configuration $\mathcal{C}_r$. 
Hence, placing $v_i$ and $c_1$ costs at most $1$ in forward orientation. 
In reverse orientation, inserting both costs at most $2+2\pi$ by Lemmas~\ref{lem:defaultApproach1} and~\ref{lem:defaultApproach2}. 
Together with $d(\mathcal{C}_\ell) \leq 3\pi$, $v_i$ and $c_1$ cost at most $3+5\pi \leq 2\alpha$ as long as $\pi \le 1/7$. 

If $e^\ast$ is a mountain and $\deg_{G_i}(v_i) = 4$, we use the default drawing putting $v_i$ in pocket $e_r$. This costs at most $0$ since $e^\ast$ will be covered and its credit can be used to pay for new mountain $(v_i,\ell_i)$. Since the reverse drawing costs at most $1-\pi$ by Lemmas~\ref{lem:defaultApproach1} and~\ref{lem:defaultApproach2} and $d(\mathcal{C}_\ell) \leq 3\pi$, for $v_i$ in this scenario, it suffices to pay at most $1+2\pi \le \alpha$ as long as $\pi \le 1/3$. 

However, if $\deg_{G_i}(v_i) = 3$ and $e^\ast (= e_\ell)$ is a mountain, more work is needed. 
Using the default drawing in the forward orientation, placing $c_1$ costs $1+\pi$ while inserting $v_i$ costs $1$ as $c_1$ creates a pocket for $v_i$. We distinguish further cases concerning $e_\ell$ in reverse orientation. 

\begin{figure}[tb]
 \centering%
 	\begin{subfigure}[b]{.4\textwidth}
    \centering
    \includegraphics[page=16,scale=1]{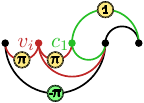}
    \subcaption{}
    \label{fig:case3b4r-1}        
    \end{subfigure}\hfil
    \begin{subfigure}[b]{.4\textwidth}
    \centering
    \includegraphics[page=17,scale=1]{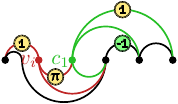}
    \subcaption{}
    \label{fig:case3b4r-2}
    \end{subfigure}
  \caption{Figures for the reverse orientation of Case~3b.4.}
  \label{fig:case3b4r}
\end{figure}

If $e_\ell$ is a pocket in reverse orientation, we place both $v_i$ and $c_1$ above it; see Fig.~\ref{fig:case3b4r-1}. 
We have to pay for the two pockets incident to $v_i$ and the mountain between $c_1$ and its rightmost neighbor which costs at most $1+\pi$  since a pocket is covered. 
Thus, since $d(\mathcal{C}_\ell) \leq 3\pi$, for $v_i$ and $c_1$, it suffices to pay at most $3+5\pi \leq 2\alpha$ as long as~$\pi \le 1/7$.

If $e_\ell$ is a mountain in reverse orientation, we further distinguish based on the type of $\mathcal{C}_r$. If $\mathcal{C}_r \in \{\mathcal{C}_a,\mathcal{C}_b\}$, we push down $e_\ell$ and place $v_i$ and $c_1$ above it; see Fig.~\ref{fig:case3b4r-2}. 
Then we have to pay for two mountains and a pocket while we can reclaim the credit on the left mountain covered by $c_1$, yielding a cost of $1+\pi$ in the reverse orientation. Again, we pay at most $3+5\pi$ for both vertices. 
If $\mathcal{C}_r = \mathcal{C}_c$, we create a new open configuration of type $\mathcal{C}_\iota$. 
If $\mathcal{C}_r \in \{\mathcal{C}_d,\mathcal{C}_g\}$, we use the default drawings for $v_i$ and $c_1$. 
Since $v_i$ covers the mountain $e_\ell$ and is placed in a pocket created by $c_1$, the credits on the edges incident to $v_i$ can be taken from covered edges and it suffices to pay for $c_1$ which costs at most $1$ as it drops into a pocket. Thus, with the debt of \OC{\ell}  we pay at most $3+4\pi$ for both vertices.
As a result, in this scenario, for both vertices it suffices to pay $3+5\pi$ which is at most $2\alpha$ for $\pi \le 1/7$.

\emph{Case 3b.5:}  $\mathcal{C}_r \in \{\mathcal{C}_f,\mathcal{C}_h\}$. 
Here, using the default drawing for $v_i$ costs at most $1+\pi$ and preserves the two stacked mountains at the left endpoint of $\mathcal{C}_r$. 
Hence, the credit of the mountain covered by $\mathcal{C}_r$ can be used to reduce the cost in the forward orientation to at most $\pi$. 
Since the reverse orientation costs at most $1+\pi$, $d(\mathcal{C}_r) \leq 2\pi$  and $d(\mathcal{C}_\ell) \leq 3\pi$, we conclude that, to insert $v_i$, it suffices to pay at most $1+7\pi \leq \alpha$ as long as $\pi \le 1/8$.

By symmetry, we discussed all cases where $\mathcal{C}_\ell \not \in \{\emptyset,\mathcal{C}_d\}$; i.e., for the last case $d(\mathcal{C}_\ell) \leq \pi$.

\emph{Case 3b.6:} $\mathcal{C}_r = \mathcal{C}_e$. 
We use default drawings for both $v_i$ and $c_1$ of configuration $\mathcal{C}_r$ for the forward orientation. 
By Lemma~\ref{lem:defaultApproach1}, inserting $v_i$ costs $1-\pi$ whereas, by construction of $\mathcal{C}_r$, inserting $c_1$ also costs $1-\pi$ yielding a total cost of $2-2\pi$ for the forward orientation.
\begin{figure}[tb]
 \centering%
 \begin{subfigure}[b]{.4\textwidth}
    \centering
    \includegraphics[page=18]{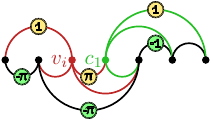}
    \subcaption{}
    \label{fig:case3b6r-1}
    \end{subfigure}\hfil
    \begin{subfigure}[b]{.4\textwidth}
    \centering
    \includegraphics[page=19]{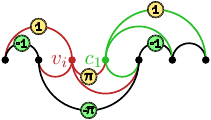}
    \subcaption{}
    \label{fig:case3b6r-2}
    \end{subfigure}
    
    \begin{subfigure}[b]{.4\textwidth}
    \centering
    \includegraphics[page=23]{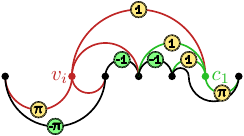}
    \subcaption{}
    \label{fig:case3b6r-3}
    \end{subfigure}\hfil
    \begin{subfigure}[b]{.4\textwidth}
    \centering
    \includegraphics[page=21]{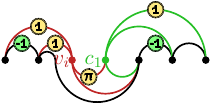}
    \subcaption{}
    \label{fig:case3b6r-4}
    \end{subfigure}
  \caption{Figures for the reverse orientation of Case~3b.6.}
  \label{fig:case3b6r}
\end{figure}

In reverse orientation, if $e^\ast$ is a pocket, we put both $v_i$ and $c_1$ in it; see Figs.~\ref{fig:case3b6r-1} to~\ref{fig:case3b6r-2}. If $\deg_{G_i}(v_i) = 4$ and $e^\ast$ is not a pocket but the leftmost covered edge by $v_i$ is, we keep the original drawing of \OC{r} but place $v_i$ in the pocket; see Fig.~\ref{fig:case3b6r-3}. Otherwise, we push down $e^\ast$ and put both $c_1$ and $v_i$ there; see Fig.~\ref{fig:case3b6r-4}, paying for the mountain from $c_1$ to its rightmost neighbor with the credit on the left mountain covered by $c_1$. 
Additionally, we pay for pocket $(v_i,c_1)$ and edge $(v_i,\ell_i)$ which may be a pocket or a mountain. 
If $e^\ast$ is a mountain and $\deg_{G_i}(v_i) = 4$, we pay for another mountain, however, we also cover another mountain whose credit we can claim. 
Hence, we may pay at most $1+\pi$ in the reverse orientation. 

Since $d(\mathcal{C}_\ell) \leq \pi$, for both $v_i$ and $c_1$, it suffices to pay at most $3 \leq 2\alpha$ as long as $\pi \le 1/2$. 

\section{SAT Formulation}
\label{app:sat}

We built a SAT formulation to check if an input graph can
be realized as a (monotone) biarc diagram with at most a  given number of biarcs $\kappa$.
Our implementation and formulation are based on an 
implementation to compute book embeddings~\cite{DBLP:conf/gd/Bekos0Z15}, with
the following changes:
\begin{itemize}
\item[--] The SAT formula was restricted to two pages.
\item[--] For each edge $e=(u,v)$, we insert both a dummy vertex $d_{e}$ representing the bend and variables $\beta_e^i$ into the SAT formula for $1 \leq i \leq \kappa$. We may enforce $d_e$ to be positioned in between $u$ and $v$ in order to compute a monotone biarc diagram. Further, $\beta_e^i=\texttt{true}$ indicates that $e$ is the $i$-th biarc the SAT solver chooses to create. Note that $\beta_e^i$ can be true for more than one $i$.
\item[--] Each edge $e=(u,v)$ must be assigned to a page only if it is not  a biarc. Otherwise, we enforce its two \emph{half-edges} $(u,d_e)$ and $(v,d_e)$ to be assigned to the different pages.
\item[--] Intersections are only checked for edges that are actually assigned to a page.
\end{itemize}

The SAT formulation yields a formula of size $O(n^3)$ for a graph on $n$
vertices. Thus, its computational use is limited to small values of $n$. For the
case where $G$ is a \emph{Kleetope} of a triangulation $T$, i.e., $T$
plus a vertex in each face that is connected to the three vertices on the
boundary of the face, we can reduce the problem size somewhat. With our implementation, we verified that no Kleetope on $n'=3n-4$ vertices
derived from a triangulation of  $n\le 14$ vertices needs more than
$\lfloor (n'-8)/3 \rfloor$ biarcs even if the outerface is prescribed.

\subsection*{Vertex Ordering}

We enumerate vertices from 1 to n and edges from 1 to m. For each edge $e_k$, we insert a dummy vertex with index $n+k$. For each pair of vertices $v_i,v_j$ (including dummy vertices), there is a variable $\sigma_{i,j}$ that is true if and only if $v_i$ appears before $v_j$ on the spine.

We ensure a isomorphic mapping of vertices (including dummy vertices) to $n+m$ positions on the spine by adding constraints which 
\begin{enumerate}
\item require vertex $v_i$ to be located either before or after vertex $v_j$ for $1 \leq i < j \leq n+m$:
\begin{eqnarray}
\left(\sigma_{i,j} \vee \sigma_{j,i}\right) \wedge \left(\neg \sigma_{i,j} \vee \neg \sigma_{j,i}\right) & \forall 1 \leq i < j \leq n+m
\end{eqnarray}
\item ensure transitivity in the sense that if $v_i$ is located before $v_j$ and $v_j$ before $v_k$ that then also $v_i$ is located before $v_k$ on the spine:
\begin{eqnarray}
\left(\neg \sigma_{i,j} \vee \neg \sigma_{j,k} \vee \sigma_{i,k}\right) \wedge \left( \sigma_{i,j} \vee \sigma_{j,k} \vee \neg \sigma_{i,k}\right) & \forall 1 \leq i < j < k \leq n+m 
\end{eqnarray} 
\end{enumerate}

\subsection*{Monotone Biarcs}

The SAT formulation can be extended easily to ensure monotonicity of biarcs. We add two constraints for each edge $e_k=(v_i,v_j)$  that enforce to place the dummy vertex of $(v_i,v_j)$ in between $v_i$ and $v_j$ on the spine:
\begin{eqnarray}
 \left(\sigma_{i,n+k} \vee \sigma_{j,n+k}\right) \wedge \left(\neg \sigma_{i,n+k} \vee \neg \sigma_{j,n+k}\right) & \forall e_k=(v_i,v_j) \in E 
\end{eqnarray}
Conversely, if we do not add these constraints, we are able to compute optimal solutions for non-monotone biarc diagrams.

\subsection*{Bounding the Number of Biarcs}

Additionally, there exists variables $\beta_i^j$ for $1 \leq i \leq m$ and $1 \leq j \leq \kappa$. Edge $e_i$ will be drawn as a biarc if and only if $\beta_i^j$ is true for any $1 \leq j \leq \kappa$. We ensure that there are at most $\kappa$ biarcs with the following constraints:

\begin{eqnarray}
\bigwedge \limits_{1 \leq i < k \leq m}\left(\neg \beta_i^j \vee \neg \beta_k^j\right) & \forall 1 \leq j \leq \kappa 
\end{eqnarray}

Note that this already creates $O(n^3)$ clauses.

\subsection*{Page Assignment}

Next, we let the SAT solver assign edges and \emph{half edges} to pages. For each edge $e_k=(v_i,v_j)$ with $i < j$, there exist 2 half edges $h_{k,1} = (v_i,v_{n+k})$ and $h_{k,2} = (v_j,v_{n+k})$. Additionally, we also call $h_{k,0} = e_k$. We have page variables $\phi_{k,i,j}$ for $1 \leq k \leq m$, $0\leq i \leq 2$ and $1\leq j \leq 2$ which are true if and only if $h_{k,i}$ is located on page $j$.

Here we have the following constraints:


\begin{eqnarray}
\left(\phi_{k,0,1} \vee \phi_{k,0,2} \vee \phi_{k,1,1} \vee \phi_{k,1,2}\right) & \forall 1 \leq k \leq m &  \text{(Assign edge or biarc 1)}\\
\left(\phi_{k,0,1} \vee \phi_{k,0,2} \vee \phi_{k,2,1} \vee \phi_{k,2,2}\right) & \forall 1 \leq k \leq m &  \text{(Assign edge or biarc 2)}\\
\left(\phi_{k,0,1} \vee \phi_{k,0,2} \vee \bigvee \limits_{1 \leq j \leq \kappa}\beta_k^j\right)  & \forall 1 \leq k \leq m & \text{(Each edge on 1 page)}  
\end{eqnarray}

If we decide that edge $e_k$ is drawn as an biarc, we additionally require its two half edges $h_{k,1}$ and~$h_{k,2}$ to be drawn on different pages while we require half edge~$h_{k,0}$ to not be drawn:

\begin{eqnarray}
\left(\neg \beta_k^j \vee \neg \phi_{k,1,1} \vee \neg \phi_{k,2,1}\right)  & \forall 1 \leq k \leq m & \forall 1 \leq j \leq \kappa  \\
\left(\neg \beta_k^j \vee \neg \phi_{k,1,2} \vee \neg \phi_{k,2,2}\right)  & \forall 1 \leq k \leq m & \forall 1 \leq j \leq \kappa  \\
\left(\neg \beta_k^j \vee \neg \phi_{k,0,1}\right)  & \forall 1 \leq k \leq m & \forall 1 \leq j \leq \kappa  \\
\left(\neg \beta_k^j \vee \neg \phi_{k,0,2}\right)  & \forall 1 \leq k \leq m & \forall 1 \leq j \leq \kappa  
\end{eqnarray}

\subsection*{Planarity}

Intersections can only occur between half edges on the same page which also have to exist due to the corresponding edge being (not) drawn as a biarc. We have a variable $\chi_{k,k^\prime,i,i^\prime}$ for $1 \leq k < k^\prime \leq n$ and $i,i^\prime \in \{0,1,2\}$ which is true if both half edges $h_{k,i}$ and $h_{k^\prime,i^\prime}$ exist and are drawn on the same page.


\begin{eqnarray}
\bigwedge \limits_{1 \leq l \leq 2} \left(\neg \phi_{k,i,l} \vee \neg \phi_{k^\prime,i^\prime,l} \vee \chi_{k,k^\prime,i,i^\prime}\right) & \forall i,i^\prime \in \{0,1,2\} \forall 1 \leq k < k^\prime \leq m 
\end{eqnarray}

Finally, we ensure that half edges $h_{k,i}=(v_{s(k,i)},v_{t(k,i)})$ and $h_{k^\prime,i^\prime}=(v_{s(k^\prime,i^\prime)}, \allowbreak v_{t(k^\prime,i^\prime)})$ do not intersect by adding the following clauses for all half edges that do not share an endpoint:
\begin{eqnarray}
\left(\neg \chi_{k,k^\prime,i,i^\prime} \vee \sigma_{s(k,i),s(k^\prime,i^\prime)} \vee \sigma_{s(k^\prime,i^\prime),t(k,i)} \vee \sigma_{t(k,i),t(k^\prime,i^\prime)}\right) \\
\left(\neg \chi_{k,k^\prime,i,i^\prime} \vee \neg \sigma_{s(k,i),s(k^\prime,i^\prime)} \vee \neg \sigma_{s(k^\prime,i^\prime),t(k,i)} \vee \neg \sigma_{t(k,i),t(k^\prime,i^\prime)}\right)  \\
\left(\neg \chi_{k,k^\prime,i,i^\prime} \vee \sigma_{t(k,i),s(k^\prime,i^\prime)} \vee \sigma_{s(k^\prime,i^\prime),s(k,i)} \vee \sigma_{s(k,i),t(k^\prime,i^\prime)}\right) \\
\left(\neg \chi_{k,k^\prime,i,i^\prime} \vee  \neg \sigma_{t(k,i),s(k^\prime,i^\prime)} \vee \neg \sigma_{s(k^\prime,i^\prime),s(k,i)} \vee \neg \sigma_{s(k,i),t(k^\prime,i^\prime)}\right) \\
\left(\neg \chi_{k,k^\prime,i,i^\prime} \vee \sigma_{s(k,i),t(k^\prime,i^\prime)} \vee \sigma_{t(k^\prime,i^\prime),t(k,i)} \vee \sigma_{t(k,i),s(k^\prime,i^\prime)}\right) \\
\left(\neg \chi_{k,k^\prime,i,i^\prime} \vee \neg \sigma_{s(k,i),t(k^\prime,i^\prime)} \vee \neg \sigma_{t(k^\prime,i^\prime),t(k,i)} \vee \neg \sigma_{t(k,i),s(k^\prime,i^\prime)}\right) \\
\left(\neg \chi_{k,k^\prime,i,i^\prime} \vee \sigma_{t(k,i),t(k^\prime,i^\prime)} \vee \sigma_{t(k^\prime,i^\prime),s(k,i)} \vee \sigma_{s(k,i),s(k^\prime,i^\prime)}\right)\\
\left(\neg \chi_{k,k^\prime,i,i^\prime} \vee  \neg \sigma_{t(k,i),t(k^\prime,i^\prime)} \vee \neg \sigma_{t(k^\prime,i^\prime),s(k,i)} \vee \neg \sigma_{s(k,i),s(k^\prime,i^\prime)}\right)  
\end{eqnarray}

\subsection*{Kleetopes}

For Kleetopes, we further decrease the size of the SAT instance to speed up
the computation time by encoding the triangulation and treating the
additional vertices as follows: In order to connect an additional vertex $v_f$
to the three vertices forming the triangle $f=(e_{k_1},e_{k_2},e_{k_3}) \in F$, $f$ only has to allow $v_f$ access to the spine. This is true
if any edge bounding $f$ is a biarc or if not all three edges
$e_{k_1},e_{k_2},e_{k_3}$ are mountains (or pockets). We encode this as follows:

\begin{eqnarray}
\bigwedge \limits_{1 \leq l \leq 2} \left(\neg \phi_{k_1,0,l} \vee \neg \phi_{k_2,0,l} \vee \neg \phi_{k_3,0,l} \vee \bigvee_{\substack{1 \leq j \leq \kappa \\ i \in \{1,2,3\}}} \beta_{k_i}^j\right) & \forall (e_{k_1},e_{k_2},e_{k_3})\in F  
\end{eqnarray}

\subsection*{Only Up-Down Biarcs}

In some applications of monotone biarc diagrams it is crucial that all biarcs have the same shape, that is, each biarc is routed above the spine left of  its spine crossing and below the spine right of its spine crossing (or vice versa). Let $e_k=(v_i,v_j)$ be an edge such that $i < j$. Hence, half-edge $h_{k,1}$ is incident to $v_i$ whereas half-edge $h_{k,2}$ is incident to $v_j$. If $e_k$ is a biarc, we assign the half-edge incident to the left endpoint of $e_k$ to page $1$:

\begin{eqnarray}
\bigwedge \limits_{1 \leq l \leq \kappa} \left(\neg \beta_{k}^l \vee \neg \sigma_{i,j} \vee \phi_{k,1,1}\right) \left(\neg \beta_{k}^l \vee \sigma_{i,j} \vee \phi_{k,2,1}\right)&  \forall e_k=(v_i,v_j) \in E
\end{eqnarray}

This requires the second half-edge to be drawn on page $2$. The linear order variable $\sigma_{i,j}$ indicates which vertex is the left endpoint of $e_k$, hence, we can chose the appropriate half-edge to be drawn on page $1$.

\subsection*{Prescribing an Outer Face}

So far, we allowed the SAT solver to chose the outer face for the output drawing. While in book embeddings and non-monotone biarc diagrams it is easy to see that the outer face can be chosen arbitrarily, it is not clear if this also holds for monotone biarc diagrams. Moreover, by identifying the facial cycle $F$ of a triangulation $T$ with a face of a given graph $G$, we can chose a face of $G$ that we want to identify with $F$. This allows us to prescribe an outer face, for $G$ even in the normal setting where the outer face might be chosen by the layout algorithm (except for one copy of $G$ if we insert a copy of $G$ in all facial cycles of $T$). Therefore, we also added the following constraints to ensure a specific outer face.

Let $f_0 = (v_i,v_j,v_k)$ be the face of $G$ that we would like to be drawn as the outer face. In order to ensure that $f_0$ is the outermost face, we have to ensure that $f_0$'s leftmost  vertex is drawn to the left of each other vertex in the graph. Similarly, the rightmost vertex of $f_0$ has to be drawn to the right of each other vertex in the graph. We formulate that as follows:

\begin{eqnarray}
\left(\neg \sigma_{i,j} \vee \neg \sigma_{i,k} \vee \sigma_{i,l} \right) \wedge  \left(\neg \sigma_{j,i} \vee \neg \sigma_{k,i} \vee \sigma_{l,i} \right) & \forall v_l\in V\setminus\{v_i,v_j,v_k\}  \\
\left(\neg \sigma_{j,i} \vee \neg \sigma_{j,k} \vee \sigma_{j,l} \right) \wedge  \left(\neg \sigma_{i,j} \vee \neg \sigma_{k,j} \vee \sigma_{l,j} \right) & \forall v_l\in V\setminus\{v_i,v_j,v_k\}  \\
\left(\neg \sigma_{k,i} \vee \neg \sigma_{k,j} \vee \sigma_{k,l} \right) \wedge  \left(\neg \sigma_{i,k} \vee \neg \sigma_{j,k} \vee \sigma_{l,k} \right) & \forall v_l\in V\setminus\{v_i,v_j,v_k\}  
\end{eqnarray}

In addition, we require the edges $f_0$ to be drawn such that all remaining vertices will only have access to the spine if they are drawn inside $f_0$: The edge between the leftmost and the rightmost vertex is drawn on page $1$, the remaining two edges are drawn on page $2$. Let $e_{ij}=(v_i,v_j)$, $e_{ik}=(v_i,v_k)$ and $e_{jk}=(v_j,v_k)$. Then:

\begin{eqnarray}
\left(\neg \sigma_{i,j} \vee \neg \sigma_{j,k} \vee \phi_{ik,0,1} \right) \wedge  \left(\neg \sigma_{i,j} \vee \neg \sigma_{j,k} \vee \phi_{ij,0,2} \right) \wedge \left(\neg \sigma_{i,j} \vee \neg \sigma_{j,k} \vee \phi_{jk,0,2} \right)~~~\\
\left(\neg \sigma_{j,i} \vee \neg \sigma_{k,j} \vee \phi_{ik,0,1} \right) \wedge  \left(\neg \sigma_{j,i} \vee \neg \sigma_{k,j} \vee \phi_{ij,0,2} \right) \wedge \left(\neg \sigma_{j,i} \vee \neg \sigma_{k,j} \vee \phi_{jk,0,2} \right)~~~\\
\left(\neg \sigma_{j,k} \vee \neg \sigma_{k,i} \vee \phi_{ij,0,1} \right) \wedge  \left(\neg \sigma_{j,k} \vee \neg \sigma_{k,i} \vee \phi_{jk,0,2} \right) \wedge \left(\neg \sigma_{j,k} \vee \neg \sigma_{k,i} \vee \phi_{ik,0,2} \right)~~~\\
\left(\neg \sigma_{k,j} \vee \neg \sigma_{i,k} \vee \phi_{ij,0,1} \right) \wedge  \left(\neg \sigma_{k,j} \vee \neg \sigma_{i,k} \vee \phi_{jk,0,2} \right) \wedge \left(\neg \sigma_{k,j} \vee \neg \sigma_{i,k} \vee \phi_{ik,0,2} \right)~~~\\
\left(\neg \sigma_{k,i} \vee \neg \sigma_{i,j} \vee \phi_{jk,0,1} \right) \wedge  \left(\neg \sigma_{k,i} \vee \neg \sigma_{i,j} \vee \phi_{ik,0,2} \right) \wedge \left(\neg \sigma_{k,i} \vee \neg \sigma_{i,j} \vee \phi_{ij,0,2} \right)~~~\\
\left(\neg \sigma_{i,k} \vee \neg \sigma_{j,i} \vee \phi_{jk,0,1} \right) \wedge  \left(\neg \sigma_{i,k} \vee \neg \sigma_{j,i} \vee \phi_{ik,0,2} \right) \wedge \left(\neg \sigma_{i,k} \vee \neg \sigma_{j,i} \vee \phi_{ij,0,2} \right)~~~\
\end{eqnarray}

\end{document}